\documentclass[11pt]{article}

\usepackage[T1]{fontenc}
\usepackage[margin=1in]{geometry}
\usepackage{amsmath,amssymb,amsthm}
\usepackage[hidelinks]{hyperref}

\newtheorem{theorem}{Theorem}[section]
\newtheorem{lemma}[theorem]{Lemma}
\newtheorem{corollary}[theorem]{Corollary}

\newcommand{\R}{\mathbb R}
\newcommand{\E}{\mathbb E}
\newcommand{\CH}{\mathcal C_H}
\newcommand{\CI}{\mathcal C_\infty}
\newcommand{\Reg}{\operatorname{Reg}}
\newcommand{\normD}[1]{\left\lVert #1\right\rVert_D}
\newcommand{\normDual}[1]{\left\lVert #1\right\rVert_{D,*}}
\DeclareMathOperator*{\argmin}{arg\,min}

\title{Optimal Regret for Online Storage Control\\
via Cumulative Policies}
\author{Kamiar Asgari \and Michael J. Neely}
\date{}

\begin{document}

\maketitle

\begin{abstract}
We study online control of a scalar storage system with adversarial
nonnegative arrivals, known retention coefficient, and convex costs
depending on both state and action. Each action must respect current
resource availability and is chosen before the current arrival and
cost function are revealed. For the existing simplex disturbance-action
policy class, we give an exact reparameterization by cumulative
allocation fractions and a decay-weighted projected subgradient update.
The resulting regret bound is independent of policy memory length.
For fixed retention coefficient and cost constants, the controller
achieves $O(\sqrt T)$ regret against the best fixed infinite-memory
policy in this class, using $O(\log T)$ memory and arithmetic operations
per round and one cost-subgradient query. A storage-specific block
construction gives a matching lower bound against every causal feasible
controller, including randomized controllers. Writing
$\tau=(1-\alpha)^{-1}$, the minimax expected regret is
$\Theta(\sqrt T\min\{T,\tau\}^{3/2})$ for every finite-memory
simplex policy class and its infinite-memory extension, when
$\alpha\in[1/2,1)$, $T\ge4$, and the positive cost constants are fixed.
This identifies the joint horizon and retention-time dependence for
these policy benchmarks.
\end{abstract}

\section{Introduction}

Storage constraints couple online decisions across time. A withdrawal
reduces the resource available for future decisions, while storage
decay limits how long past decisions affect the system. When arrivals
and costs are unknown in advance, a controller must learn a useful
allocation rule while maintaining feasibility at every round.

The benchmark is the best fixed policy in a prescribed feasible
class, evaluated on its own state trajectory under the same inputs.
This trajectory comparison is standard in adversarial online
control~\cite{agarwal2019online}; it accounts for the effect of a
policy's past actions on its present costs. Here we use the simplex
disturbance-action control (SDAC) class introduced for storage by
Asgari and Neely~\cite{asgari2026sdac}. For fixed model parameters,
their entropic update gives $O(\sqrt{T\log(H+1)})$ regret for memory
$H$ and $O(\sqrt{T\log\log T})$ for the infinite-memory extension.
Their independent-sign lower bound leaves a gap in retention-time
dependence. We address these two gaps for the same policy benchmarks.

Our first contribution is a regret bound uniform in the policy memory
length $H$ (Theorem~\ref{thm:upper}). Cumulative allocation fractions
give a bijection with the SDAC parameters. In these coordinates, a
weighted Euclidean geometry controls both surrogate-cost optimization
and the tracking error caused by changing policies. The total weight
is $\alpha\sum_{j=0}^{T-2}\alpha^j$, independently of $H$.
For fixed model parameters, the resulting bound is $O(\sqrt T)$,
removing the factor $\sqrt{\log(H+1)}$ while preserving the comparator
class and feasibility at every round. A standard weighted isotonic
projection implements the update in $O(H)$ operations.
Geometric truncation then gives $O(\sqrt T)$ regret against the
infinite-memory class with $O(\log T)$ memory and arithmetic per round
for fixed $\alpha$, excluding the cost-subgradient query
(Corollary~\ref{cor:infinite}).

Our second contribution is a storage-specific lower bound that
matches the upper bound in both horizon and retention time
(Theorem~\ref{thm:minimax}). Let $\tau=(1-\alpha)^{-1}$.
For every finite-memory SDAC class with $H\ge1$ and its
infinite-memory extension, the minimax expected regret has order
\[
\sqrt T\min\{T,\tau\}^{3/2}
=
\min\{T^2,\tau^{3/2}\sqrt T\}
\]
when $\alpha\in[1/2,1)$, $T\ge4$, and $L_x,L_u>0$ are fixed.
The construction uses only binary arrivals, affine state costs, and
two comparator policies, and applies to every causal feasible
controller, including randomized controllers. It identifies both the
$\tau^{3/2}\sqrt T$ regime and the $T^2$ regime when retention time
is at least the horizon.

These statements concern the known scalar dynamics, full cost feedback,
and absence of active capacity saturation specified in
Section~\ref{sec:model}. The comparator is fixed within the SDAC class;
the minimax infimum allows all causal feasible controllers. The
logarithmic implementation guarantee fixes $\alpha$, whereas the joint
minimax theorem permits $\alpha$ to vary and uses an exact
$T-1$-memory representation for the infinite-memory benchmark.

\subsection{Related work}

The closest predecessor, Asgari and Neely~\cite{asgari2026sdac},
already provides the feasible SDAC class, online clipping, inclusion
of fixed-fraction policies, geometric truncation, and the
independent-sign lower bound. We retain these ingredients. The advances
are the cumulative-coordinate bound uniform in $H$ and the lower
bound establishing the joint dependence on $T$ and $\tau$.
In particular, optimal $\sqrt T$ dependence for fixed finite memory
was already established there.

Agarwal et al.~\cite{agarwal2019online} develop disturbance-action
policies for linear systems with adversarial disturbances and convex
costs, obtaining regret guarantees relative to strongly stable linear
controllers through convex surrogate losses. This provides the broader
algorithmic setting for SDAC. Faster horizon rates are possible with
additional cost structure: Foster and
Simchowitz~\cite{foster2020logarithmic} obtain logarithmic regret for
known quadratic state and control costs under adversarial disturbances.
Our minimax statement instead ranges over adversarial convex costs
with the stated Lipschitz and subgradient bounds, and its lower bound
already holds for affine costs.

Online convex optimization with memory formalizes the dependence of
current losses on earlier decisions. Anava, Hazan, and
Mannor~\cite{anava2015memory} study finite memory and policy regret.
Kumar, Dean, and Kleinberg~\cite{kumar2023unbounded} introduce an
effective-memory quantity, prove matching worst-case bounds for their
general framework, and derive log-free $O(\sqrt T)$ regret for online
linear control. Their analysis also uses weighted norms to exploit
decaying influence. Thus log-free regret and decay-sensitive geometry
already have precedents. Here the policy length $H$ and physical
retention time $\tau$ play different roles: our bound is uniform in
$H$ but necessarily grows with $\tau$. The lower bound realizes this
dependence within the resource-feasible storage model itself, rather
than over an unrestricted family of losses with memory.

Hard constraints have also been studied in online control.
Li, Das, and Li~\cite{li2021onlineConstrained} enforce affine state
and action constraints through disturbance-action policies and
constraint tightening, assuming the existence of a strictly safe
linear controller.
The inequality $u_t\le\alpha x_t$ here is itself affine in the joint
state--action pair. The useful additional structure is resource
nonnegativity: the SDAC parameter restrictions guarantee fixed-policy
feasibility, and clipping enforces availability along the online
trajectory without an interior safety margin. Liu, Yang, and
Ying~\cite{liu2023constraints} give cumulative-violation bounds for
adversarial constraints and anytime guarantees for static constraints,
including zero violation for affine constraints via projection. Golowich
et al.~\cite{golowich2024population} obtain regret guarantees for
population dynamics on the simplex. Their simplex describes the
population state, whereas our simplex describes allocation fractions
assigned to arrivals of different ages.

Computational dependence on stability is a separate issue from
retention dependence in regret. Brahmbhatt
et al.~\cite{brahmbhatt2026spectral} approximate diagonalizably stable
linear policies using spectral filters, with amortized runtime
dependence polylogarithmic in the inverse stability margin. Our
cumulative representation is exact for the SDAC class, and the
implementation retains an explicit arrival buffer. Its memory choice
scales as $O(\min\{T,1+\tau\log T\})$, becoming $O(\log T)$ when
$\alpha$ is fixed. The contribution here is the sharp storage regret
bound and a simple feasible implementation, rather than an improvement
in general control's runtime dependence on stability.

Storage control also has stochastic and competitive formulations.
Qin et al.~\cite{qin2016storage} use Lyapunov optimization for
capacity- and rate-constrained storage, with feasibility and long-run
suboptimality guarantees under stochastic inputs. Kim
et al.~\cite{kim2016storage} analyze online energy-storage management
competitively, including networked storage and decay. These works
address storage constraints under different benchmarks and information
patterns. Our performance criterion is additive regret relative to a
fixed feasible policy's complete trajectory; our model has no active
capacity saturation. Thus adversarial storage control and decay are
established settings, while the result here identifies the minimax
rate for the specified SDAC benchmark and feedback model.

For energy-harvesting communication, Shaviv and
\"Ozg\"ur~\cite{shaviv2016universally} give near-optimal fixed-fraction
power policies for independent and identically distributed arrivals.
Arafa et al.~\cite{arafa2018online} extend fixed-fraction
analysis to general concave utilities. Such policies motivate an
important subclass of our infinite-memory comparator. Asgari and
Neely~\cite{asgari2020bregman} study online convex optimization with
energy-harvesting constraints and a regret--battery-capacity tradeoff
under independent and identically distributed arrivals. Their actions
can use the current arrival, and their fixed-action benchmark is
constrained by mean energy supply. Here actions precede the current
arrival, both arrivals and costs can be adversarial, and regret compares
feasible policy trajectories with stock-dependent costs.

The optimization steps use established methods: projected online
subgradient descent~\cite{zinkevich2003online}, quadratic mirror
descent~\cite{beck2003mirror}, and weighted isotonic regression via
the pool-adjacent-violators algorithm~\cite{best1990isotonic}.
Our analysis chooses their coordinates and weights to control both
policy optimization and the storage tracking error.

\section{Model and feedback}
\label{sec:model}

We write $\R_+=[0,\infty)$, use column vectors, and set
$\langle v,z\rangle=v^\top z$. The symbols $\mathbf0$ and
$\mathbf1$ denote the all-zero and all-one vectors of the required
dimension. All logarithms are natural, and $0^0=1$ in the
geometric sequences below.

Fix a horizon $T\ge2$ and a known retention coefficient
$\alpha\in(0,1)$. Let $L_x,L_u>0$ be known bounds on the cost
variation, specified below. The system starts empty and evolves according to
\begin{equation}
\label{eq:dynamics}
x_1=0,
\qquad
x_{t+1}=\alpha x_t-u_t+w_t,
\qquad t=1,\ldots,T.
\end{equation}
Here $x_t$ is the resource available at the beginning of round $t$,
$u_t$ is the withdrawal, and $w_t\in[0,1]$ is the arrival.
Every action must satisfy
\begin{equation}
\label{eq:feasibility}
0\le u_t\le\alpha x_t.
\end{equation}

At round $t$, the controller observes $x_t$ and chooses $u_t$
using only previously revealed information and its internal
randomness. The current arrival $w_t$ and the entire cost function
$c_t$ are then revealed. The controller incurs $c_t(x_t,u_t)$,
and the state evolves according to \eqref{eq:dynamics}.

Define
\begin{equation}
\label{eq:constants}
\tau=\frac{1}{1-\alpha},
\qquad
\tau_T=\sum_{j=0}^{T-2}\alpha^j,
\qquad
C=L_x+L_u,
\qquad
K=L_x+\alpha L_u.
\end{equation}
Feasibility and $0\le w_t\le1$ give
$0\le x_{t+1}\le\alpha x_t+1$. Induction from $x_1=0$ therefore gives
\[
0\le x_t\le\sum_{j=0}^{t-2}\alpha^j
=\frac{1-\alpha^{t-1}}{1-\alpha}\le\tau.
\]
Also, $1\le\tau_T\le\min\{T-1,\tau\}$ and $K\le C$.
Thus all feasible state--action pairs lie in
\[
\mathcal D_\alpha
=
\{(x,u):0\le x\le\tau,\ 0\le u\le\alpha x\}.
\]
The dynamics have no active capacity saturation. A physical capacity
at least $\tau$ suffices to ensure this property.

Each cost $c_t:\R_+^2\to\R$ is convex and satisfies
\begin{equation}
\label{eq:lipschitz}
|c_t(x,u)-c_t(x',u')|
\le L_x|x-x'|+L_u|u-u'|
\end{equation}
on $\mathcal D_\alpha$, for fixed constants $L_x,L_u>0$.
After observing $c_t$, the controller can obtain, at every point
$(x,u)\in\mathcal D_\alpha$, including its boundary, a subgradient
\begin{equation}
\label{eq:oracle}
(\zeta_t^x,\zeta_t^u)\in\partial c_t(x,u),
\qquad
|\zeta_t^x|\le L_x,
\qquad
|\zeta_t^u|\le L_u.
\end{equation}
Here a subgradient $g$ of a convex function $f$ at $z$ satisfies
$f(z')\ge f(z)+\langle g,z'-z\rangle$ for every point $z'$ in
its domain. In particular, the queried cost subgradient satisfies
\begin{equation}
\label{eq:cost-support}
c_t(x',u')\ge c_t(x,u)
+\zeta_t^x(x'-x)+\zeta_t^u(u'-u),
\qquad (x',u')\in\mathcal D_\alpha.
\end{equation}
The bounded query in \eqref{eq:oracle} is an explicit assumption
also at boundary points of $\mathcal D_\alpha$.

The online cumulative cost is
\[
J_T^{\mathrm{on}}=\sum_{t=1}^T c_t(x_t,u_t).
\]
All comparator policies are evaluated from the same empty initial
state and on the same arrival and cost sequence.
The upper bounds below hold for every admissible realized sequence.
The lower bounds use sequences fixed before play.

We set $w_s=0$ for $s\le0$ and interpret empty sums as zero.
When $T=1$, every feasible controller and comparator has
$(x_1,u_1)=(0,0)$, so regret is zero.

\section{Cumulative policy classes}
\label{sec:policies}

For an integer $H\ge1$, define
\begin{equation}
\label{eq:policy-set}
\CH
=
\{y\in\R^H:0\le y_1\le\cdots\le y_H\le1\}.
\end{equation}
Set $y_0=0$ and extend a finite policy by $y_i=y_H$ for $i>H$.
The action generated by a fixed policy $y\in\CH$ is
\begin{equation}
\label{eq:policy-action}
u_t(y)=\sum_{i=1}^H\alpha^i(y_i-y_{i-1})w_{t-i}.
\end{equation}
The same parameter vector is used throughout the fixed-policy
trajectory.

The coordinate $y_i$ is the cumulative allocation fraction through
age $i$. In particular, an arrival $w_s$ contributes
$\alpha^i(y_i-y_{i-1})w_s$ to the withdrawal at round $s+i$.
No further withdrawals are assigned to that arrival after age $H$.

The simplex disturbance-action class of
\cite{asgari2026sdac} is
\[
\mathcal M_H
=
\left\{m\in\R_+^H:\sum_{i=1}^H m_i\le1\right\}.
\]
The maps
\begin{equation}
\label{eq:bijection}
y_i=\sum_{j=1}^i m_j,
\qquad
m_i=y_i-y_{i-1}
\end{equation}
give a bijection between $\mathcal M_H$ and $\CH$.
For clarity, the SDAC policy associated with $m\in\mathcal M_H$ is
\[
u_t(m)=\sum_{i=1}^H\alpha^i m_iw_{t-i},
\qquad
x_1(m)=0,\qquad
x_{t+1}(m)=\alpha x_t(m)-u_t(m)+w_t.
\]
If $m\in\mathcal M_H$, its partial sums satisfy
$0\le y_1\le\cdots\le y_H\le1$. Conversely, $y\in\CH$ gives
$m_i\ge0$ and $\sum_i m_i=y_H\le1$. These maps are inverses.
Substituting $m_i=y_i-y_{i-1}$ makes the two action formulas
identical. The common initial condition and dynamics then make
their states identical by induction, so comparator costs agree exactly.

\begin{lemma}[Trajectory formula and feasibility]
\label{lem:trajectory}
For every $y\in\CH$,
\begin{align}
\label{eq:policy-state}
x_t(y)
&=\sum_{j=1}^{t-1}\alpha^{j-1}(1-y_{j-1})w_{t-j},\\
\label{eq:policy-residual}
\alpha x_t(y)-u_t(y)
&=\sum_{j=1}^{t-1}\alpha^j(1-y_j)w_{t-j}\ge0.
\end{align}
Consequently, the policy is feasible, and both $x_t(y)$ and
$u_t(y)$ are affine functions of $y$.
\end{lemma}

\begin{proof}
The state formula holds at $t=1$. Assuming it holds at $t$,
subtracting \eqref{eq:policy-action} from $\alpha x_t(y)$ gives
\eqref{eq:policy-residual}. Adding $w_t$ gives the state formula
at $t+1$. This proves the formulas by induction.
The increments $y_i-y_{i-1}$ are nonnegative, so $u_t(y)\ge0$.
Since $y_j\le1$, the residual is nonnegative.
Affineness follows directly from the formulas.
\end{proof}

Two policies used later have $y_i^{(0)}=0$ and $y_i^{(1)}=1$
for $i\ge1$; both have zeroth coordinate zero. They respectively
make no withdrawals and
withdraw each arrival's surviving amount one round later:
\begin{equation}
\label{eq:endpoints}
\begin{aligned}
x_t(y^{(0)})
&=\sum_{j=1}^{t-1}\alpha^{j-1}w_{t-j},
&
u_t(y^{(0)})&=0,\\
x_t(y^{(1)})&=w_{t-1},
&
u_t(y^{(1)})&=\alpha w_{t-1}.
\end{aligned}
\end{equation}

The infinite-memory class is
\begin{equation}
\label{eq:infinite-class}
\CI
=
\{(y_i)_{i\ge0}:y_0=0,\ 0\le y_1\le y_2\le\cdots\le1\},
\end{equation}
with actions
\[
u_t(y)=\sum_{i=1}^{t-1}\alpha^i(y_i-y_{i-1})w_{t-i}.
\]
Lemma~\ref{lem:trajectory} remains valid because every sum at a
fixed time is finite. The correspondence \eqref{eq:bijection}
also gives a bijection with
\[
\mathcal M_\infty
=\{(m_i)_{i\ge1}:m_i\ge0,\ \sum_{i\ge1}m_i\le1\},
\qquad
u_t(m)=\sum_{i=1}^{t-1}\alpha^im_iw_{t-i}.
\]
Indeed, for $y\in\CI$ the nonnegative increments satisfy
$\sum_{i=1}^n m_i=y_n$, so their total mass is
$\lim_{n\to\infty}y_n\le1$. Conversely, partial sums of any
$m\in\mathcal M_\infty$ give a member of $\CI$.
The action and state identities hold at every round because
only finitely many coordinates are used.

Fixed-fraction policies are widely used in energy-harvesting
control~\cite{shaviv2016universally,arafa2018online}.
As in the SDAC representation~\cite{asgari2026sdac}, the
infinite-memory class includes every feasible fixed-fraction
policy $u_t=kx_t$, $k\in[0,\alpha]$. Indeed, take
\[
y_i=1-\left(1-\frac{k}{\alpha}\right)^i,
\qquad i\ge1.
\]
The state formula becomes
\[
x_t(y)=\sum_{j=1}^{t-1}(\alpha-k)^{j-1}w_{t-j},
\]
and the action formula gives $u_t(y)=kx_t(y)$.

For $\mathcal C=\CH$ or $\mathcal C=\CI$, define
\begin{equation}
\label{eq:regret}
J_T(y)=\sum_{t=1}^T c_t(x_t(y),u_t(y)),
\qquad
\Reg_T(\mathcal C)
=
J_T^{\mathrm{on}}-\inf_{y\in\mathcal C}J_T(y).
\end{equation}
Only $y_1,\ldots,y_{T-1}$ can affect costs through round $T$.
Thus $\CI$ and $\mathcal C_{T-1}$ have identical comparator
costs on this horizon. Likewise, a finite memory larger than
$T-1$ can be replaced by $T-1$ without changing the benchmark.

\section{Weighted cumulative-policy control}
\label{sec:algorithm}

We now take $1\le H\le T-1$.

\subsection{Arrival coefficients and weighted geometry}

Define
\begin{equation}
\label{eq:features}
\begin{aligned}
r_t&=\sum_{j=1}^{t-1}\alpha^{j-1}w_{t-j},\\
a_{t,i}&=\alpha^i w_{t-i-1},
&&1\le i<H,\\
b_{t,i}&=\alpha^i w_{t-i}-\alpha^{i+1}w_{t-i-1},
&&1\le i<H,\\
a_{t,H}&=\sum_{j=H+1}^{t-1}\alpha^{j-1}w_{t-j},\\
b_{t,H}&=\alpha^H w_{t-H}.
\end{aligned}
\end{equation}
These quantities depend only on arrivals observed before round $t$.
Collect the coefficients in vectors $a_t,b_t\in\R^H$.

Choose weights
\begin{equation}
\label{eq:weights}
d_i=\alpha^i\quad(1\le i<H),
\qquad
d_H=\sum_{j=H}^{T-1}\alpha^j,
\qquad
W_T=\sum_{i=1}^H d_i=\alpha\tau_T.
\end{equation}
Let $D=\operatorname{diag}(d_1,\ldots,d_H)$ and define
\[
\normD{z}^2=\sum_{i=1}^H d_i z_i^2,
\qquad
\normDual{g}^2=\sum_{i=1}^H\frac{g_i^2}{d_i}.
\]
For these norms, ordinary Cauchy--Schwarz applied to
$(g_i/\sqrt{d_i})_i$ and $(\sqrt{d_i}z_i)_i$ gives
\begin{equation}
\label{eq:weighted-duality}
|\langle g,z\rangle|
\le
\left(\sum_i\frac{g_i^2}{d_i}\right)^{1/2}
\left(\sum_i d_i z_i^2\right)^{1/2}
=\normDual{g}\normD{z}.
\end{equation}
All weights are positive. The last coordinate receives the entire
remaining decay weight because $y_H$ also determines the stock
remaining at ages beyond $H$. The total weight $W_T$ does not
depend on $H$.

For a fixed policy, define the surrogate cost
\begin{equation}
\label{eq:surrogate}
\ell_t(y)=c_t(x_t(y),u_t(y)).
\end{equation}
This evaluates round $t$ as if the same policy $y$ had been used
from the beginning.

\begin{lemma}[Coefficient and subgradient bounds]
\label{lem:features}
For every $t\le T$ and $y\in\CH$,
\begin{equation}
\label{eq:affine}
x_t(y)=r_t-a_t^\top y,
\qquad
u_t(y)=b_t^\top y.
\end{equation}
Moreover,
\begin{equation}
\label{eq:coefficient-bounds}
0\le a_{t,i}\le d_i,
\qquad
|b_{t,i}|\le d_i,
\qquad
\normDual{a_t},\normDual{b_t}\le\sqrt{W_T}.
\end{equation}
Consequently, for $y,z\in\CH$,
\begin{equation}
\label{eq:sensitivity}
\begin{aligned}
|x_t(y)-x_t(z)|&\le\sqrt{W_T}\normD{y-z},\\
|u_t(y)-u_t(z)|&\le\sqrt{W_T}\normD{y-z}.
\end{aligned}
\end{equation}
The surrogate $\ell_t$ is convex. A cost-subgradient query at
$(x_t(y),u_t(y))$ gives
\begin{equation}
\label{eq:gradient}
g_t(y)=-\zeta_t^x a_t+\zeta_t^u b_t
\in\partial\ell_t(y),
\qquad
\normDual{g_t(y)}\le C\sqrt{W_T}.
\end{equation}
Moreover,
\begin{equation}
\label{eq:surrogate-lipschitz}
|\ell_t(y)-\ell_t(z)|
\le C\sqrt{W_T}\normD{y-z},
\qquad y,z\in\CH.
\end{equation}
\end{lemma}

\begin{proof}
Grouping the terms in \eqref{eq:policy-state}, with $y_i=y_H$
for $i>H$, gives
\[
x_t(y)
=
r_t-\sum_{i=1}^{H-1}\alpha^i w_{t-i-1}y_i
-y_H\sum_{j=H+1}^{t-1}\alpha^{j-1}w_{t-j}.
\]
Expanding successive differences in \eqref{eq:policy-action}
gives
\[
u_t(y)
=
\sum_{i=1}^{H-1}
(\alpha^i w_{t-i}-\alpha^{i+1}w_{t-i-1})y_i
+\alpha^H w_{t-H}y_H.
\]
These are \eqref{eq:affine}.

For $i<H$, both $w_{t-i}$ and $\alpha w_{t-i-1}$ lie in
$[0,1]$, giving the coordinate bounds. For the last coordinate,
\[
0\le a_{t,H}\le\sum_{j=H}^{T-2}\alpha^j\le d_H,
\qquad
|b_{t,H}|\le\alpha^H\le d_H.
\]
Thus, for $v=a_t$ or $v=b_t$,
\[
\normDual{v}^2
=
\sum_i\frac{v_i^2}{d_i}
\le\sum_i d_i=W_T.
\]
Weighted Cauchy--Schwarz proves \eqref{eq:sensitivity}.

Write $F_t(y)=(x_t(y),u_t(y))$. By \eqref{eq:affine},
$F_t(\lambda y+(1-\lambda)z)=\lambda F_t(y)+(1-\lambda)F_t(z)$.
For $\lambda\in[0,1]$, cost convexity therefore gives
\[
\ell_t(\lambda y+(1-\lambda)z)
\le\lambda\ell_t(y)+(1-\lambda)\ell_t(z).
\]
Applying \eqref{eq:cost-support} at $F_t(y)$ yields, for $z\in\CH$,
\begin{align*}
\ell_t(z)-\ell_t(y)
&\ge \zeta_t^x\bigl(x_t(z)-x_t(y)\bigr)
+\zeta_t^u\bigl(u_t(z)-u_t(y)\bigr)\\
&=(-\zeta_t^x a_t+\zeta_t^u b_t)^\top(z-y),
\end{align*}
which proves the subgradient formula. Furthermore,
\[
\normDual{g_t(y)}
\le|\zeta_t^x|\normDual{a_t}
+|\zeta_t^u|\normDual{b_t}
\le(L_x+L_u)\sqrt{W_T}.
\]
Finally, \eqref{eq:lipschitz} and \eqref{eq:sensitivity} give
\[
|\ell_t(y)-\ell_t(z)|
\le L_x|x_t(y)-x_t(z)|+L_u|u_t(y)-u_t(z)|
\le C\sqrt{W_T}\normD{y-z}.
\]
\end{proof}

\subsection{Controller}

We apply projected online subgradient
descent~\cite{zinkevich2003online} in the $D$ norm, equivalently
mirror descent with quadratic regularizer
$y^\top Dy/2$~\cite{beck2003mirror}.
Fix a step size $\eta>0$ and initialize
$y_1=\frac12\mathbf1$. At round $t$, compute
\begin{equation}
\label{eq:controller-action}
\widehat x_t=r_t-a_t^\top y_t,
\qquad
v_t=b_t^\top y_t,
\qquad
u_t=\min\{v_t,\alpha x_t\}.
\end{equation}
The pair $(\widehat x_t,v_t)$ is the feasible state--action
pair of the current fixed policy. The actual state $x_t$ was
generated by earlier policies and may differ from $\widehat x_t$.
The minimum enforces current resource availability.

After observing $w_t$ and $c_t$, query a cost subgradient at
$(\widehat x_t,v_t)$, set
\[
g_t=-\zeta_t^x a_t+\zeta_t^u b_t,
\]
and update
\begin{equation}
\label{eq:update}
y_{t+1}
=
\argmin_{y\in\CH}
\left\{
\eta g_t^\top y+\frac12\normD{y-y_t}^2
\right\}.
\end{equation}
The set $\CH$ is nonempty, compact, and convex, and $D$ is
positive definite. Thus the update has a unique minimizer.
For $z_t=y_t-\eta D^{-1}g_t$, completing the square gives
\[
\eta g_t^\top y+\frac12\normD{y-y_t}^2
=\frac12\normD{y-z_t}^2
+\eta g_t^\top y_t-\frac{\eta^2}{2}\normDual{g_t}^2.
\]
The last two terms do not depend on $y$, so the update is
exactly the weighted projection of $z_t$ onto $\CH$.

The controller uses one cost-subgradient query per round.
Full cost feedback and \eqref{eq:oracle} make the query at the
hypothetical pair available.
Lemma~\ref{lem:trajectory} gives $v_t\ge0$, so
\eqref{eq:controller-action} ensures feasibility. All information
used to choose the action is available before the current
arrival and cost are revealed.

The coefficients can be maintained using the last $H$ arrivals
and the recursions
\begin{equation}
\label{eq:recursions}
r_{t+1}=\alpha r_t+w_t,
\qquad
\rho_{t+1}=\alpha\rho_t+\alpha^H w_{t-H},
\end{equation}
where $r_1=\rho_1=0$ and $\rho_t=a_{t,H}$.
The old arrival $w_{t-H}$ is used before inserting $w_t$ into
the arrival buffer.

\subsection{Weighted projection}

The projection is a bounded weighted isotonic regression problem.
We use the classical pool-adjacent-violators
algorithm~\cite{best1990isotonic}, followed by clipping to $[0,1]$.
The following lemma records the bounded version and its linear
arithmetic complexity for completeness.

\begin{lemma}[Bounded weighted isotonic projection]
\label{lem:projection}
Let $d_i>0$ and $z\in\R^H$. Start with one block per coordinate.
Whenever a left block has a larger weighted mean than the
next block, merge them. A block $B$ has weight and mean
\[
d_B=\sum_{i\in B}d_i,
\qquad
\mu_B=\frac{\sum_{i\in B}d_i z_i}{d_B}.
\]
At termination, let $p_i$ equal the mean of its block, and set
\[
\widehat p_i=\min\{1,\max\{0,p_i\}\}.
\]
Then $\widehat p$ is the unique minimizer of
\begin{equation}
\label{eq:projection}
\min_{0\le y_1\le\cdots\le y_H\le1}
\frac12\sum_{i=1}^H d_i(y_i-z_i)^2.
\end{equation}
A stack implementation uses $O(H)$ arithmetic operations and memory.
\end{lemma}

\begin{proof}
First omit the bounds zero and one. Every block produced by the
merging procedure satisfies
\begin{equation}
\label{eq:prefix}
\sum_{i\in P}d_i(z_i-\mu_B)\ge0
\end{equation}
for each prefix $P$ of the block, with equality for the entire
block. This is immediate for a singleton.

Suppose adjacent blocks $A,B$ with $\mu_A>\mu_B$ are merged
to form a block of mean $\mu$. For a prefix contained in $A$,
the new residual sum is nonnegative because $\mu\le\mu_A$.
For a prefix consisting of $A$ followed by a prefix $P$ of $B$,
the residual sum is at least
\[
d_A(\mu_A-\mu)+d_P(\mu_B-\mu)
\ge d_A(\mu_A-\mu)+d_B(\mu_B-\mu)=0.
\]
Here $d_P\le d_B$ and $\mu_B-\mu<0$.
Thus \eqref{eq:prefix} is preserved.

At termination, define
\[
\lambda_0=\lambda_H=0,
\qquad
\lambda_k=\sum_{i=1}^k d_i(z_i-p_i),
\quad 1\le k<H.
\]
The prefix property gives $\lambda_k\ge0$. At a block boundary
$\lambda_k=0$; within a block $p_k=p_{k+1}$. Hence
\[
d_i(p_i-z_i)+\lambda_i-\lambda_{i-1}=0,
\qquad
\lambda_i(p_i-p_{i+1})=0.
\]
These identities directly prove optimality. For any nondecreasing
vector $q$, summation by parts and the complementary identities give
\[
\langle D(p-z),q-p\rangle
=\sum_{i=1}^{H-1}\lambda_i
\bigl((q_{i+1}-q_i)-(p_{i+1}-p_i)\bigr)
=\sum_{i=1}^{H-1}\lambda_i(q_{i+1}-q_i)\ge0.
\]
For $F(q)=\frac12\normD{q-z}^2$, expanding the square gives
\[
F(q)-F(p)
=\langle D(p-z),q-p\rangle+\frac12\normD{q-p}^2\ge0.
\]
Hence $p$ minimizes the unbounded isotonic problem.

To impose the bounds, take any $y\in\CH$.
The map $q\mapsto q-\min\{1,\max\{0,q\}\}$ is nondecreasing,
so $y+p-\widehat p$ is nondecreasing. Optimality of $p$ yields
\[
\langle D(p-z),y-\widehat p\rangle\ge0.
\]
Also, coordinatewise clipping gives
\[
(\widehat p_i-p_i)(y_i-\widehat p_i)\ge0.
\]
Adding these inequalities gives
\[
\langle D(\widehat p-z),y-\widehat p\rangle\ge0,
\]
which proves optimality in \eqref{eq:projection}.
Strict convexity gives uniqueness.

A stack inserts each coordinate once and performs at most
$H-1$ merges. Its arithmetic and memory costs are therefore
$O(H)$.
\end{proof}

Combining this projection with \eqref{eq:features} and
\eqref{eq:recursions}, the controller uses $O(H)$ arithmetic
operations and memory per round, apart from the cost-subgradient
query.

\section{Regret guarantees}
\label{sec:upper}

\begin{theorem}[Memory-independent regret bound]
\label{thm:upper}
Under the assumptions of Section~\ref{sec:model}, the controller
\eqref{eq:controller-action}--\eqref{eq:update} is causal and
feasible. For every $\eta>0$ and every admissible realized sequence,
\begin{equation}
\label{eq:upper-eta}
\Reg_T(\CH)
\le
\frac{W_T}{8\eta}
+\eta T W_T\left(\frac{C^2}{2}+CK\tau_T\right).
\end{equation}
With
\begin{equation}
\label{eq:eta}
\eta=\frac{1}{2\sqrt{T(C^2+2CK\tau_T)}},
\end{equation}
the bound becomes
\begin{equation}
\label{eq:upper}
\Reg_T(\CH)
\le U_T
:=
\frac{\alpha\tau_T}{2}
\sqrt{T(C^2+2CK\tau_T)}.
\end{equation}
In particular,
\begin{equation}
\label{eq:upper-simple}
\Reg_T(\CH)
\le
\frac{\sqrt3}{2}\alpha(L_x+L_u)\tau_T^{3/2}\sqrt T
\le
\frac{\sqrt3\alpha(L_x+L_u)}
{2(1-\alpha)^{3/2}}\sqrt T.
\end{equation}
None of these bounds depends on $H$.
\end{theorem}

The proof separates optimization of the fixed-policy surrogate
costs from tracking the state generated by changing policies.
This follows the surrogate-and-movement approach used in online
learning with memory and disturbance-action
control~\cite{anava2015memory,agarwal2019online,kumar2023unbounded}.
The storage-specific step is to bound both terms in the same
cumulative-coordinate geometry, uniformly in $H$.

\begin{lemma}[Optimization and policy movement]
\label{lem:optimization}
The update \eqref{eq:update} satisfies
\begin{equation}
\label{eq:movement}
\normD{y_{t+1}-y_t}
\le\eta\normDual{g_t}
\le\eta C\sqrt{W_T}.
\end{equation}
For every $y^*\in\CH$,
\begin{equation}
\label{eq:surrogate-regret}
\sum_{t=1}^T
\bigl(\ell_t(y_t)-\ell_t(y^*)\bigr)
\le
\frac{W_T}{8\eta}
+\frac{\eta TC^2W_T}{2}.
\end{equation}
\end{lemma}

\begin{proof}
For any $y\in\CH$, the segment
$y_{t+1}+s(y-y_{t+1})$, $0\le s\le1$, lies in $\CH$.
The right derivative at $s=0$ of the minimized objective is
nonnegative. Consequently,
\begin{equation}
\label{eq:update-optimality}
\langle \eta g_t+D(y_{t+1}-y_t),y-y_{t+1}\rangle
\ge0,
\qquad y\in\CH.
\end{equation}
With $\Delta_t=y_{t+1}-y_t$, taking $y=y_t$ gives
\[
\normD{\Delta_t}^2
\le-\eta g_t^\top\Delta_t
\le\eta\normDual{g_t}\normD{\Delta_t}.
\]
If $\Delta_t=0$, the movement bound is immediate; otherwise,
division by $\normD{\Delta_t}$ proves its first inequality.
Lemma~\ref{lem:features} gives the second.

Taking $y=y^*$ in \eqref{eq:update-optimality} and expanding
squared norms yields
\[
\eta g_t^\top(y_{t+1}-y^*)
\le\frac12\bigl(
\normD{y_t-y^*}^2-\normD{y_{t+1}-y^*}^2
-\normD{\Delta_t}^2\bigr).
\]
Also, weighted Cauchy--Schwarz and
$ab\le(a^2+b^2)/2$ imply
\[
\eta g_t^\top(y_t-y_{t+1})
\le\eta\normDual{g_t}\normD{\Delta_t}
\le\frac{\eta^2}{2}\normDual{g_t}^2
+\frac12\normD{\Delta_t}^2.
\]
Adding and dividing by $\eta$ gives
\[
g_t^\top(y_t-y^*)
\le
\frac{\normD{y_t-y^*}^2-\normD{y_{t+1}-y^*}^2}{2\eta}
+\frac{\eta}{2}\normDual{g_t}^2.
\]
Convexity bounds
$\ell_t(y_t)-\ell_t(y^*)$ by the left-hand side.
Summing over $t$, using
$\normDual{g_t}^2\le C^2W_T$, and observing that
\[
\normD{y_1-y^*}^2
=
\sum_i d_i\left(\frac12-y_i^*\right)^2
\le\frac{W_T}{4}
\]
proves \eqref{eq:surrogate-regret}.
\end{proof}

\begin{lemma}[Tracking a changing policy]
\label{lem:tracking}
For any sequence $y_t\in\CH$ used in
\eqref{eq:controller-action}, define
\[
e_t=|x_t-x_t(y_t)|.
\]
Then $e_1=0$ and
\begin{equation}
\label{eq:tracking-step}
e_{t+1}
\le\alpha e_t+\sqrt{W_T}\normD{y_{t+1}-y_t},
\qquad t<T.
\end{equation}
Consequently,
\begin{equation}
\label{eq:tracking-loss}
\sum_{t=1}^T
|c_t(x_t,u_t)-\ell_t(y_t)|
\le
K\sqrt{W_T}\tau_T
\sum_{t=1}^{T-1}\normD{y_{t+1}-y_t}.
\end{equation}
For the update \eqref{eq:update}, the right-hand side is at most
$\eta TCKW_T\tau_T$.
\end{lemma}

\begin{proof}
Write $v_t=u_t(y_t)$ and $[q]_+=\max\{q,0\}$.
The actual and current fixed-policy transitions are
\[
x_{t+1}=[\alpha x_t-v_t]_++w_t,
\qquad
x_{t+1}(y_t)=[\alpha x_t(y_t)-v_t]_++w_t.
\]
The second identity uses fixed-policy feasibility.
The positive-part map is $1$-Lipschitz, so
\[
|x_{t+1}-x_{t+1}(y_t)|\le\alpha e_t.
\]
Combining this with \eqref{eq:sensitivity} proves
\eqref{eq:tracking-step}.

Set $\delta_s=\normD{y_{s+1}-y_s}$. Induction from $e_1=0$
gives $e_t\le\sqrt{W_T}\sum_{s=1}^{t-1}\alpha^{t-1-s}\delta_s$.
Changing the order of the finite sums yields
\begin{align*}
\sum_{t=1}^T e_t
&\le\sqrt{W_T}\sum_{s=1}^{T-1}\delta_s
       \sum_{t=s+1}^T\alpha^{t-1-s}\\
&\le\sqrt{W_T}\tau_T\sum_{s=1}^{T-1}\delta_s.
\end{align*}
Furthermore,
\[
0\le v_t-u_t
=[v_t-\alpha x_t]_+
\le[\alpha x_t(y_t)-\alpha x_t]_+
\le\alpha e_t.
\]
Hence \eqref{eq:lipschitz} implies
\[
|c_t(x_t,u_t)-\ell_t(y_t)|
\le(L_x+\alpha L_u)e_t=Ke_t.
\]
Summing proves \eqref{eq:tracking-loss}.
Finally, \eqref{eq:movement} gives
\[
K\sqrt{W_T}\tau_T
\sum_{t=1}^{T-1}\normD{y_{t+1}-y_t}
\le\eta TCKW_T\tau_T.
\]
\end{proof}

\begin{proof}[Proof of Theorem~\ref{thm:upper}]
Causality and feasibility were established in
Section~\ref{sec:algorithm}. For any fixed $y^*\in\CH$,
\begin{align*}
J_T^{\mathrm{on}}-J_T(y^*)
&=
\sum_{t=1}^T
\bigl(\ell_t(y_t)-\ell_t(y^*)\bigr)
+
\sum_{t=1}^T
\bigl(c_t(x_t,u_t)-\ell_t(y_t)\bigr)\\
&\le
\frac{W_T}{8\eta}
+\frac{\eta TC^2W_T}{2}
+\eta TCKW_T\tau_T.
\end{align*}
Here we used Lemmas~\ref{lem:optimization} and
\ref{lem:tracking}.
A minimizing comparator exists because $\CH$ is compact and
$J_T$ is continuous. Choosing it proves \eqref{eq:upper-eta}.

The step size \eqref{eq:eta} balances the inverse-step-size
term with the terms proportional to the step size, yielding
\eqref{eq:upper}. Finally, $K\le C$ and $\tau_T\ge1$ imply
\[
C^2+2CK\tau_T\le3C^2\tau_T.
\]
Substituting this and $\tau_T\le\tau$ proves
\eqref{eq:upper-simple}.
\end{proof}

\subsection{Infinite-memory comparison}

We use the geometric-tail comparison for SDAC
policies~\cite{asgari2026sdac}, stated below in cumulative coordinates
with the finite-horizon factor $\tau_T$.
For $y\in\CI$, define its truncation by
\begin{equation}
\label{eq:truncation}
y_i^{[H]}=y_{\min\{i,H\}},
\qquad
\varepsilon_{H,T}(y)
=
\sum_{i=H+1}^{T-1}\alpha^i(y_i-y_{i-1}).
\end{equation}
The truncated policy preserves allocations through age $H$
and omits all later allocations.

\begin{lemma}[Truncation cost]
\label{lem:truncation}
For every $y\in\CI$,
\begin{equation}
\label{eq:truncation-cost}
|J_T(y^{[H]})-J_T(y)|
\le
T(L_u+L_x\tau_T)\varepsilon_{H,T}(y),
\end{equation}
where
\begin{equation}
\label{eq:tail}
0\le\varepsilon_{H,T}(y)
\le\alpha^{H+1}(1-y_H)\le\alpha^{H+1}.
\end{equation}
If $H\ge T-1$, the two trajectories agree through round $T$.
\end{lemma}

\begin{proof}
Let
\[
q_t=u_t(y)-u_t(y^{[H]})
=
\sum_{i=H+1}^{t-1}\alpha^i(y_i-y_{i-1})w_{t-i}.
\]
Then $0\le q_t\le\varepsilon_{H,T}(y)$.
The state difference
$z_t=x_t(y^{[H]})-x_t(y)$ satisfies
\[
z_1=0,
\qquad
z_{t+1}=\alpha z_t+q_t.
\]
Thus
\[
0\le z_t
=\sum_{s=1}^{t-1}\alpha^{t-1-s}q_s
\le\tau_T\varepsilon_{H,T}(y).
\]
Applying \eqref{eq:lipschitz} and summing proves
\eqref{eq:truncation-cost}.

If $H<T-1$, monotonicity gives
\[
\varepsilon_{H,T}(y)
\le
\alpha^{H+1}
\sum_{i=H+1}^{T-1}(y_i-y_{i-1})
\le\alpha^{H+1}(1-y_H).
\]
If $H\ge T-1$, all omitted-action sums are empty, so
$q_t=z_t=0$ through round $T$.
\end{proof}

\begin{corollary}[Infinite-memory regret]
\label{cor:infinite}
With the step size \eqref{eq:eta}, for every $y\in\CI$,
\begin{equation}
\label{eq:individual-infinite}
J_T^{\mathrm{on}}-J_T(y)
\le
U_T+T(L_u+L_x\tau_T)\varepsilon_{H,T}(y).
\end{equation}
Consequently,
\begin{equation}
\label{eq:infinite-regret}
\Reg_T(\CI)
\le
\begin{cases}
U_T+(L_u+L_x\tau_T)\alpha^{H+1}T,
&1\le H<T-1,\\
U_T,
&H=T-1.
\end{cases}
\end{equation}
In particular, choose
\begin{equation}
\label{eq:log-memory}
H_T=
\min\left\{
T-1,\
\max\left\{
1,\
\left\lceil\frac{\log T}{\log(1/\alpha)}\right\rceil
\right\}
\right\}.
\end{equation}
Then
\begin{equation}
\label{eq:log-regret}
\Reg_T(\CI)
\le U_T+\alpha(L_u+L_x\tau_T).
\end{equation}
For fixed $\alpha,L_x,L_u$, this is $O(\sqrt T)$ regret with
$O(\log T)$ arithmetic operations and memory per round,
apart from the cost-subgradient query.
\end{corollary}

\begin{proof}
Since $y^{[H]}\in\CH$, Theorem~\ref{thm:upper} and
Lemma~\ref{lem:truncation} give
\[
J_T^{\mathrm{on}}-J_T(y)
=
J_T^{\mathrm{on}}-J_T(y^{[H]})
+J_T(y^{[H]})-J_T(y)
\le
U_T+T(L_u+L_x\tau_T)\varepsilon_{H,T}(y).
\]
Taking the supremum over $y\in\CI$ proves
\eqref{eq:infinite-regret}.
If $H_T<T-1$, then $T\alpha^{H_T}\le1$.
If $H_T=T-1$, the comparison is exact.
These cases give \eqref{eq:log-regret}.
The complexity follows from $H_T=O(\log T)$ for fixed $\alpha$.
\end{proof}

\section{Lower bounds and the minimax rate}
\label{sec:lower}

Throughout this section, $\mathcal C$ denotes either $\CH$,
for any integer $H\ge1$, or $\CI$. Define
\begin{equation}
\label{eq:minimax}
\mathfrak R_T(\mathcal C)
=
\inf_{\mathcal A}
\sup_{(w_t,c_t)_{t=1}^T}
\E_{\mathcal A}[\Reg_T(\mathcal C)].
\end{equation}
The infimum is over all causal feasible controllers, including
randomized controllers. The supremum is over admissible
deterministic sequences fixed before the run. The expectation
is over the controller's internal randomness.

Both constructions below compare with the two policies in
\eqref{eq:endpoints}, which belong to every class under consideration.
We use the sharp Khintchine inequality
\cite{haagerup1981khintchine}: for independent uniform signs
$\xi_i\in\{-1,1\}$ and deterministic real coefficients $a_i$,
\begin{equation}
\label{eq:khintchine}
\E\left|\sum_i a_i\xi_i\right|
\ge
\frac1{\sqrt2}\left(\sum_i a_i^2\right)^{1/2}.
\end{equation}

The following elementary facts make the probability arguments explicit.

\begin{lemma}[Unrevealed signs and deterministic instances]
\label{lem:averaging}
Let $\mathcal F$ be a sigma-field, let $V$ be an integrable
$\mathcal F$-measurable random variable, and let $\xi$ be a random
sign with $\E[\xi\mid\mathcal F]=0$. Then $\E[\xi V]=0$.
Also, let $\omega$ range over a finite family of deterministic
input sequences, sampled independently of a fixed controller's
randomization. If
$\E_\omega\E_{\mathcal A}[\Reg_T(\mathcal C;\omega)]\ge r$,
then some fixed sequence $\omega$ satisfies
$\E_{\mathcal A}[\Reg_T(\mathcal C;\omega)]\ge r$.
\end{lemma}

\begin{proof}
Conditional expectation gives
\[
\E[\xi V]
=\E\bigl[\E[\xi V\mid\mathcal F]\bigr]
=\E\bigl[V\E[\xi\mid\mathcal F]\bigr]=0.
\]
For the second claim, the average of the finitely many values
$\E_{\mathcal A}[\Reg_T(\mathcal C;\omega)]$ cannot exceed their
maximum. A maximizing sequence is fixed before play, although
its selection may depend on the controller.
\end{proof}

In the randomized constructions below, expectations include both
the sampled signs and the controller's randomness until a sign
sequence is fixed. The signs are independent of that randomness.
All variables to which the first part of Lemma~\ref{lem:averaging}
is applied are bounded, because $x_t\le\tau$ and the horizon is finite.
The lemma converts each averaged lower bound into a deterministic
oblivious instance.

\subsection{A block construction}

Independent random signs held fixed within blocks expose the
cost of memory; see, for example, the lower bounds of
Kumar et al.~\cite{kumar2023unbounded}.
Our construction uses recharge, clearing, holding, and reward phases
to realize a retention-dependent separation under storage feasibility.
It requires no state resets and uses the two endpoint policies
in \eqref{eq:endpoints} as comparators.

\begin{theorem}[Block lower bound]
\label{thm:block}
Let $b\ge1$ be an integer and $T\ge3b+1$. Define
\[
S_b=\frac{1-\alpha^b}{1-\alpha},
\qquad
B=3b+1,
\qquad
N=\left\lfloor\frac TB\right\rfloor.
\]
Then
\begin{equation}
\label{eq:block-bound}
\mathfrak R_T(\mathcal C)
\ge
\frac{L_x\alpha^{b+1}S_b^2}{4\sqrt2}\sqrt N.
\end{equation}
The lower bound holds with binary arrivals and affine costs
depending only on the state.
\end{theorem}

\begin{proof}
Fix a causal feasible controller. Draw independent uniform
signs $\xi_1,\ldots,\xi_N\in\{-1,1\}$, independently of its
randomization. Partition the first $NB$ rounds into blocks.
Block $n$ consists of the following successive phases:
\[
\begin{array}{c|c|c|c}
\text{Phase}&\text{Rounds}&w_t&c_t(x,u)\\ \hline
\text{Recharge}&b&1&0\\
\text{Clearing}&1&0&0\\
\text{Holding}&b&0&(L_x/2)\alpha^b x\\
\text{Reward}&b&0&-(L_x/2)(1-\xi_n)x
\end{array}
\]
All remaining rounds have zero arrival and zero cost.
Every realization is admissible: state slopes lie in
$[-L_x,L_x]$, and action slopes are zero.

The policy $y^{(1)}$ satisfies $x_t(y^{(1)})=w_{t-1}$.
Its state is therefore zero throughout every holding and reward
phase, giving
\begin{equation}
\label{eq:one-cost}
J_T(y^{(1)})=0.
\end{equation}
No reset is imposed on the online state or on the state of the
no-withdrawal comparator.

Let
\[
\sigma_n=(n-1)B+2b+2
\]
be the first reward round of block $n$.
For any feasible trajectory, denote its block cost by $C_n$ and set
\[
V_n=\frac{L_x}{2}S_bx_{\sigma_n}.
\]
There are no arrivals during the holding and reward phases, so
$0\le x_{t+1}\le\alpha x_t$. For $j=0,\ldots,b-1$,
\begin{equation}
\label{eq:block-stock}
\alpha^b x_{\sigma_n-b+j}
\ge\alpha^j x_{\sigma_n}
\ge x_{\sigma_n+j}.
\end{equation}
Since $1-\xi_n\ge0$,
\begin{align}
\label{eq:block-pathwise}
C_n
&=
\frac{L_x}{2}\sum_{j=0}^{b-1}
\left(
\alpha^b x_{\sigma_n-b+j}
-(1-\xi_n)x_{\sigma_n+j}
\right)\nonumber\\
&\ge
\frac{L_x}{2}\sum_{j=0}^{b-1}
\left(
\alpha^j x_{\sigma_n}
-(1-\xi_n)\alpha^j x_{\sigma_n}
\right)
=\xi_nV_n.
\end{align}
Equality holds when no withdrawals are made in these two phases.

The online value $V_n$ is determined before the first reward
cost reveals $\xi_n$. Conditional on the controller's random
seed and all information available before that revelation,
$\xi_n$ has mean zero. Consequently,
\begin{equation}
\label{eq:online-nonnegative}
\E[J_T^{\mathrm{on}}]
\ge\sum_{n=1}^N\E[\xi_nV_n]=0.
\end{equation}

For the no-withdrawal policy $y^{(0)}$, define
\[
\overline V_n
=
\frac{L_x}{2}S_bx_{\sigma_n}(y^{(0)}).
\]
These coefficients are deterministic because arrivals are
deterministic and the policy ignores costs.
If $s_n=(n-1)B+1$ is the first recharge round, then
\[
x_{s_n+b}(y^{(0)})
=
\alpha^b x_{s_n}(y^{(0)})+S_b
\ge S_b.
\]
The clearing and holding phases have $b+1$ rounds with no
arrivals or withdrawals. Therefore,
\begin{equation}
\label{eq:comparator-amplitude}
x_{\sigma_n}(y^{(0)})\ge\alpha^{b+1}S_b,
\qquad
\overline V_n\ge\frac{L_x}{2}\alpha^{b+1}S_b^2.
\end{equation}
Equality in \eqref{eq:block-pathwise} gives
\[
J_T(y^{(0)})=\sum_{n=1}^N\overline V_n\xi_n,
\qquad
\E[J_T(y^{(0)})]=0.
\]

Using \eqref{eq:one-cost}, \eqref{eq:online-nonnegative},
and $-\min\{0,z\}=(|z|-z)/2$, we obtain
\begin{align*}
\E[\Reg_T(\mathcal C)]
&\ge
\E\left[
J_T^{\mathrm{on}}-\min\{0,J_T(y^{(0)})\}
\right]\\
&\ge
\frac12\E\left|\sum_{n=1}^N\overline V_n\xi_n\right|\\
&\ge
\frac1{2\sqrt2}
\left(\sum_{n=1}^N\overline V_n^2\right)^{1/2}\\
&\ge
\frac{L_x\alpha^{b+1}S_b^2}{4\sqrt2}\sqrt N.
\end{align*}
The penultimate inequality is \eqref{eq:khintchine}.
Fixing a sign sequence that attains at least this average
produces a deterministic oblivious instance.
Since the controller was arbitrary, the minimax bound follows.
\end{proof}

\begin{corollary}[Finite-horizon lower bound]
\label{cor:finite-lower}
Let $\alpha\in[1/2,1)$, $T\ge4$, and
$M=\min\{T,\tau\}$. Then
\begin{equation}
\label{eq:finite-lower}
\mathfrak R_T(\mathcal C)
\ge
\frac{L_x}{4096\sqrt2}\sqrt T\,M^{3/2}.
\end{equation}
\end{corollary}

\begin{proof}
Choose
\[
b=
\min\{\lfloor T/4\rfloor,\lfloor\tau/2\rfloor\}.
\]
Since $T\ge4$ and $\tau\ge2$,
\[
b\ge1,\qquad b\ge M/8,\qquad B=3b+1\le4b\le T.
\]
Bernoulli's inequality gives
\[
\alpha^b=(1-1/\tau)^b\ge1-b/\tau\ge\frac12.
\]
Hence
\[
\alpha^{b+1}\ge\frac14,
\qquad
S_b\ge b\alpha^b\ge\frac b2,
\qquad
N=\left\lfloor\frac TB\right\rfloor
\ge\frac{T}{2B}\ge\frac{T}{8b}.
\]
Substitution into \eqref{eq:block-bound} yields
\[
\mathfrak R_T(\mathcal C)
\ge
\frac{L_x}{4\sqrt2}\frac14\frac{b^2}{4}
\sqrt{\frac{T}{8b}}
=
\frac{L_x}{256}\sqrt T\,b^{3/2}
\ge
\frac{L_x}{4096\sqrt2}\sqrt T\,M^{3/2}.
\]
\end{proof}

\begin{corollary}[Beyond the retention time]
\label{cor:retention}
If $\alpha\in[1/2,1)$ and $T\ge10\tau$, then
\begin{equation}
\label{eq:retention-lower}
\mathfrak R_T(\mathcal C)
\ge
\frac{L_x}{128\sqrt{20}}\tau^{3/2}\sqrt T.
\end{equation}
\end{corollary}

\begin{proof}
Set
\[
q=\frac{\log2}{\log(1/\alpha)},
\qquad
b=\lceil q\rceil.
\]
Since $q\le b<q+1$ and $\alpha^q=1/2$,
\[
\frac\alpha2<\alpha^b\le\frac12,
\qquad
S_b\ge\frac\tau2,
\qquad
\alpha^{b+1}>\frac{\alpha^2}{2}\ge\frac18.
\]
Using $-\log\alpha\ge1-\alpha=1/\tau$,
\[
B=3b+1<3\tau\log2+4<5\tau.
\]
Since $T\ge10\tau>2B$,
\[
N=\left\lfloor\frac TB\right\rfloor
\ge\frac{T}{2B}>\frac{T}{10\tau}.
\]
Theorem~\ref{thm:block} gives
\[
\mathfrak R_T(\mathcal C)
\ge
\frac{L_x}{4\sqrt2}\frac18
\left(\frac\tau2\right)^2
\sqrt{\frac{T}{10\tau}}
=
\frac{L_x}{128\sqrt{20}}\tau^{3/2}\sqrt T.
\]
\end{proof}

\subsection{Action costs and small horizons}

For completeness, we also record the independent-sign lower bound
from \cite{asgari2026sdac}, expressed in cumulative coordinates.
It includes the action-cost constant $L_u$ and holds for every
$\alpha\in(0,1)$ and $T\ge2$.

\begin{theorem}[Independent-sign lower bound]
\label{thm:sign}
For $T\ge2$,
\begin{equation}
\label{eq:sign-bound}
\mathfrak R_T(\mathcal C)
\ge
\frac1{2\sqrt2}
\left(\sum_{t=2}^T\beta_t^2\right)^{1/2},
\qquad
\beta_t=
\alpha\left(
L_u+\frac{L_x(1-\alpha^{t-2})}{1-\alpha}
\right).
\end{equation}
\end{theorem}

\begin{proof}
Fix a causal feasible controller.
Draw independent uniform signs
$\xi_1,\ldots,\xi_T\in\{-1,1\}$ and set
\[
w_t=1,
\qquad
c_t(x,u)=\xi_t(L_xx-L_uu).
\]
Every realization is admissible.
The current sign is independent of the state and action chosen
before its revelation, so
\[
\E[J_T^{\mathrm{on}}]=0.
\]

For the two endpoint policies and $t\ge2$,
\[
\begin{aligned}
x_t(y^{(0)})&=\frac{1-\alpha^{t-1}}{1-\alpha},
&
u_t(y^{(0)})&=0,\\
x_t(y^{(1)})&=1,
&
u_t(y^{(1)})&=\alpha.
\end{aligned}
\]
Their expected cumulative costs are zero, and
\[
J_T(y^{(0)})-J_T(y^{(1)})
=
\sum_{t=2}^T\beta_t\xi_t.
\]
Using $\min\{a,b\}=(a+b-|a-b|)/2$ gives
\begin{align*}
\E[\Reg_T(\mathcal C)]
&\ge
\E\left[
J_T^{\mathrm{on}}
-\min\{J_T(y^{(0)}),J_T(y^{(1)})\}
\right]\\
&=
\frac12\E\left|\sum_{t=2}^T\beta_t\xi_t\right|\\
&\ge
\frac1{2\sqrt2}
\left(\sum_{t=2}^T\beta_t^2\right)^{1/2}.
\end{align*}
The last step uses \eqref{eq:khintchine}.
Fixing a sign sequence that attains at least this average
proves the result.
\end{proof}

\begin{corollary}[Long-horizon independent-sign bound]
\label{cor:sign-long}
Define
\[
t_{\mathrm{mix}}
=
\left\lceil\frac{\log2}{\log(1/\alpha)}\right\rceil,
\qquad
G_\alpha=\alpha(L_u+L_x\tau).
\]
If $T\ge2t_{\mathrm{mix}}+2$, then
\[
\mathfrak R_T(\mathcal C)\ge\frac{G_\alpha}{8}\sqrt T.
\]
\end{corollary}

\begin{proof}
For $t=t_{\mathrm{mix}}+2,\ldots,T$,
$\alpha^{t-2}\le1/2$ and hence $\beta_t\ge G_\alpha/2$.
There are at least $T/2$ such indices.
Theorem~\ref{thm:sign} therefore gives
\[
\mathfrak R_T(\mathcal C)
\ge
\frac1{2\sqrt2}
\sqrt{\frac T2\frac{G_\alpha^2}{4}}
=
\frac{G_\alpha}{8}\sqrt T.
\]
\end{proof}

For fixed positive cost constants, the independent-sign bound
has order $\tau\sqrt T$ as retention time grows and the horizon
exceeds the retention time.
The block construction strengthens this dependence to
$\tau^{3/2}\sqrt T$.

\subsection{Matching minimax rate}

\begin{theorem}[Matching horizon and retention dependence]
\label{thm:minimax}
Let $\alpha\in[1/2,1)$, $T\ge4$, and $M=\min\{T,\tau\}$.
For every finite-memory class $\CH$, $H\ge1$, and for $\CI$,
\begin{equation}
\label{eq:matching}
\frac{L_x}{4096\sqrt2}\sqrt T\,M^{3/2}
\le
\mathfrak R_T(\mathcal C)
\le
\frac{\sqrt3}{2}(L_x+L_u)\sqrt T\,M^{3/2}.
\end{equation}
Consequently, for fixed $L_x,L_u>0$,
\begin{equation}
\label{eq:rate}
\mathfrak R_T(\mathcal C)
=
\Theta\!\left(
\sqrt T\min\{T,(1-\alpha)^{-1}\}^{3/2}
\right)
=
\Theta\!\left(
\min\{T^2,(1-\alpha)^{-3/2}\sqrt T\}
\right).
\end{equation}
\end{theorem}

\begin{proof}
The lower bound is Corollary~\ref{cor:finite-lower}.
For the upper bound, set
\[
H_{\mathcal C}
=
\begin{cases}
\min\{H,T-1\},&\mathcal C=\CH,\\
T-1,&\mathcal C=\CI.
\end{cases}
\]
The comparator costs equal those of
$\mathcal C_{H_{\mathcal C}}$ through round $T$.
Running the controller with this memory and applying
Theorem~\ref{thm:upper} gives
\[
\mathfrak R_T(\mathcal C)
\le
\frac{\sqrt3}{2}\alpha(L_x+L_u)\tau_T^{3/2}\sqrt T
\le
\frac{\sqrt3}{2}(L_x+L_u)M^{3/2}\sqrt T,
\]
because $\alpha\le1$ and
$\tau_T\le\min\{T-1,\tau\}\le M$.
\end{proof}

The exact minimax upper bound for $\CI$ above uses memory $T-1$.
For fixed $\alpha$, Corollary~\ref{cor:infinite} provides the
more economical logarithmic-memory implementation with the
additional bounded truncation term in \eqref{eq:log-regret}.
The matching statement concerns dependence on $T$ and $\tau$
for fixed positive cost constants; it does not assert a sharp
joint dependence on independently varying $L_x$ and $L_u$.

\section{Conclusion}

Cumulative allocation fractions give an exact representation of
the existing SDAC class in which a decay-weighted projected update
controls optimization and tracking error uniformly in policy memory.
For fixed model parameters, this yields $O(\sqrt T)$ regret against
the best fixed infinite-memory SDAC policy with logarithmic memory
and arithmetic per round, excluding the cost-subgradient query.
The storage-specific block construction establishes the matching
rate $\Theta(\sqrt T\min\{T,\tau\}^{3/2})$ for
$\alpha\in[1/2,1)$, $T\ge4$, and fixed $L_x,L_u>0$.
The optimality statement concerns these policy benchmarks under
full cost feedback and without active capacity saturation; the
lower bound holds even when the controller is any causal feasible
algorithm.

\bibliographystyle{plain}
\bibliography{references}

\end{document}